\documentclass[letterpaper, 10 pt, conference]{ieeeconf} 
\usepackage{graphicx} 
\graphicspath{{images/}} 
\usepackage{float} 
\usepackage{hyperref} 
\usepackage{mathtools} 
\usepackage{amsmath} 
\usepackage{amssymb} 
\usepackage{comment}
\newtheorem{remark}{Remark}
\newtheorem{assumption}{Assumption}
\newtheorem{definition}{Definition}
\newtheorem{theorem}{Theorem}
\newtheorem{lemma}{Lemma}

\usepackage{xcolor} 
\usepackage{algorithm} 
\usepackage{algpseudocode} 

\newcommand{\blue}[1]{\textcolor{blue}{#1}}
\newcommand{\set}[1]{\mathcal{#1}}

\IEEEoverridecommandlockouts
\title{On Sampling Time and Invariance}

\author{Spencer Schutz$^{1}$, Charlott Vallon$^{1}$, Ben Recht$^{2}$, and Francesco Borrelli$^{1}$ 
\thanks{$^{1}$S. Schutz, C. Vallon, and F. Borrelli are with the Department of Mechanical Engineering, University of California, Berkeley,  Berkeley, CA 94720. $\{$spencer.schutz, charlottvallon, fborrelli$\}$@berkeley.edu}%
\thanks{$^{2}$B. Recht is with the Department of Electrical Engineering and Computer Science, University of California, Berkeley, Berkeley, CA 94720. brecht@berkeley.edu}
}

\begin{document}
\maketitle

\begin{abstract} 
Invariant sets define regions of the state space where system constraints are always satisfied. The majority of numerical techniques for computing invariant sets
have been developed for discrete-time systems with a fixed sampling time. Understanding how invariant sets change with sampling time is critical for designing adaptive-sampling control schemes that ensure constraint satisfaction. 
We introduce 
$M$-step hold control invariance, a generalization of traditional control invariance, and show its practical use to assess the link between control sampling frequency and constraint satisfaction. 
We robustify $M$-step hold control invariance against model mismatches and discretization errors, paving the way for adaptive-sampling control strategies.
\end{abstract}

\section{Introduction} \label{sec:intro}
Control design for high-performance, safety-critical autonomous systems requires safety guarantees. Recursive feasibility, the guarantee that constraints are satisfiable at all times, is traditionally achieved through the notion of control invariance. Control invariant sets contain the states for which a control input exists that ensures constraint satisfaction for all future times \cite{Blanchini_Set_Inv}. In optimal control formulations such as Model Predictive Control (MPC), constraining the system to a control invariant set guarantees recursive feasibility \cite{Blanchini_Set_Inv, Borrelli_MPC_book, Decardi_Rob_MPC}. While safety is desired for the real, continuous time system, digital control systems operate in discrete time using forecasts of sampled-data models obtained via discretization of continuous time dynamics. This requires choosing a sampling time, typically backed by intuition of the dynamics, task complexity, and computational resources. 

So-called ``adaptive-sampling control" varies the sampling time during system operation, leveraging the benefits of different sampling times as required by the task. 
Small sampling times enable more frequent actuation changes, providing stability during difficult or highly dynamic tasks. Large sampling times mean fewer actuation changes, which can save energy and computational resources during simple tasks and steady-state operation. 
Adaptive-sampling has been implemented in various controllers, including PID \cite{Dorf_PID_Vary_Ts}, LQR \cite{Henriksson_LQR_Vary_Ts}, and MPC \cite{Xue_MPC_Vary_Ts, Gomozov_MPC_Vary_Ts}. Sampling time is typically adapted based on characteristics of the state 
or tracking reference.

How are safety guarantees affected by changing the sampling time? Previous implementations of adaptive-sampling MPC \cite{Xue_MPC_Vary_Ts, Gomozov_MPC_Vary_Ts} did not provide recursive feasibility guarantees. While Lyapunov-based stability guarantees have been developed for adaptive-sampling systems \cite{Hu_Lyapunov_Vary_Ts}, we found no studies of invariance-based safety guarantees. Here we examine the effects of sampling time on control invariant sets. One might expect that a finer discretization (i.e., a smaller sampling time) leads to an expansion of the system's control invariant set, as the controller can react more frequently. However, we will show that this is not always the case. This observation suggests that traditional control invariance is insufficient to reason about constrained adaptive-sampling control.

Existing invariance generalizations \cite{Shen_recurrence, Olaru_p-invariance} allow constraint relaxation for particular time steps. 
Here we introduce $M$-step hold control invariance, which enforces constraints at all time steps and supports adaptive sampling schemes. Unlike methods finding trajectories that also satisfy inter-sample constraints \cite{Elango_InterSample1, Uzun_Intersample4}, we identify the entire set of states for which inter-sample constraints can be satisfied, by robustifying $M$-step hold invariance against modeling and discretization errors.
We provide algorithms and computational methods for determining these sets.


\section{Problem Formulation} \label{sec:problem}
Consider a continuous time, time invariant system model with state $x\in\mathbb{R}^n$ and control input $u\in\mathbb{R}^m$ subject to state and input constraints:
\begin{subequations}
    \begin{gather} 
        \dot{x}(t)=f_c\big(x(t),u(t)\big) \label{eqn:ct_dyn} \\
        x(t)\in\mathcal{X},~ u(t)\in\mathcal{U},~\forall t\geq0  \label{eqn:const_ct}
    \end{gather}
\end{subequations}

\begin{definition}
\label{def:ctr_inv_ct}The set $\mathcal{C}_{f_c}\subseteq\mathcal{X}$ is \textbf{\textit{control invariant}} for a model \eqref{eqn:ct_dyn} subject to constraints \eqref{eqn:const_ct} if for all states in $\mathcal{C}_{f_c}$ there exists a controller such that the constraints are satisfied for all future time: 
    \begin{gather*}
        x_0\in\mathcal{C}_{f_c} \Rightarrow \exists~ u(t)\in\mathcal{U}~ \text{s.t.}~ x(t)\in\mathcal{C}_{f_c} \  \forall t\geq0. \label{eqn:ctr_inv_ct}
    \end{gather*}
\end{definition}
$\mathcal{C}_{\infty, f_c}$ denotes the ``maximal control invariant set" for $f_c$, i.e., the control invariant set containing all other $\mathcal{C}_{f_c}\subseteq\mathcal{X}$.

Digital control systems typically make use of $f^*_{d, T_s}$, a discretization of $f_c$ with method $^*$ (e.g. Forward Euler, Runge-Kutta, \dots) and sampling time $T_s$:
\begin{subequations}
    \begin{gather} 
        x[k+1]=f^*_{d,T_s}\big(x[k],u[k]\big) \label{eqn:dt_dyn} \\
        x[k]\in\mathcal{X},~ u[k]\in\mathcal{U},~\forall k\in\mathbb{N}_0  \label{eqn:const_dt}
    \end{gather}
\end{subequations}

\begin{definition} $\mathcal{C}_{f^*_{d, T_s}}$ is \textbf{\textit{control invariant}} for \eqref{eqn:dt_dyn} subject to \eqref{eqn:const_dt} if for all initial states in $\mathcal{C}_{f^*_{d,T_s}}$ there exists a controller such that the constraints are satisfied at all time samples: \label{def:ctr_inv_dt}
    \begin{align*}
        &\forall k\in\mathbb{N}_0,~ x[k]\in\mathcal{C}_{f^*_{d, T_s}} \Rightarrow \label{eqn:ctr_inv_dt}\\
        & \exists
        u[k]\in\mathcal{U} ~:~ x[k+1]=f^*_{d,T_s}\big(x[k],u[k]\big) \in\mathcal{C}_{f^*_{d, T_s}}
        \nonumber
    \end{align*}
\end{definition}
\noindent We denote with $\mathcal{C}_{\infty, f^*_{d,T_s}}$ the \textbf{\textit{maximal control invariant set}} for $f^*_{d,T_s}$, i.e., the control invariant set containing all $\mathcal{C}_{f^*_{d,T_s}}\subseteq\mathcal{X}$. 
When the discretization technique `$*$'  is  clear from the context, we will use the simplified notation: \begin{equation}
\mathcal{C}_{\infty,f^*_{d,T_s}} \rightarrow \mathcal{C}_{\infty, T_s}, \  f^*_{d,T_s,} \rightarrow f_{T_s}. \label{eqn:sim_not}
\end{equation}

We are interested in designing 
adaptive-sampling control schemes, where the feedback controller sampling time is adjusted during system operation with guaranteed constraint satisfaction. 
For easier tasks, a lower sampling frequency might suffice, whereas a higher sampling frequency 
may be necessary for more challenging tasks. Here, task difficulty could be measured in terms of how far the system is from a certain constraint boundary.

With this in mind, we consider how to extend the guarantees of control invariance to systems with different discretization sampling times.
A natural idea is to use multiple discretizations of $f_c$ to create different control invariant sets $\mathcal{C}_{\infty,T_s}$, where the sampling time for each discrete model is an integer multiple of the smallest sampling time $T_s$:
\begin{equation}
    T_{s,M}=MT_s,~M\in\mathbb{N}_+ \label{eqn:ts_relation}
\end{equation}
For example, discretizing $f_c$ with $T_s=0.1$ and $M\in\{1,3,8\}$ generates three models $f_{0.1},~f_{0.3},\text{ and } f_{0.8}$.

We begin with the observation that maximal control invariant sets for discretized models do not change predictably when the sampling time is changed. 
Consider the system
\begin{subequations}
    \begin{gather*}
        \dot{x}(t)=f_c(\cdot,\cdot)=\begin{bmatrix} 0 & 1 \\ 0 & 0 \end{bmatrix}x(t)+\begin{bmatrix} 0 \\ 1 
        \end{bmatrix}u(t) \label{eqn:dbl_int_ct}\\
        \begin{bmatrix} -10 \\ -10 \end{bmatrix}\leq x(t)\leq\begin{bmatrix} 10 \\ 10 \end{bmatrix},~-10\leq u(t)\leq 10 \label{eqn:dbl_int_constraint}
    \end{gather*}
\end{subequations}
and two exact discretizations with $T_{s,1}=0.5$ and $M=3$ ($T_{s,3}=1.5$):
    \begin{align}
        x[k+1]&=f_{0.5}(\cdot,\cdot)=\begin{bmatrix} 1 & 0.5 \\ 0 & 1 \end{bmatrix}x[k]+\begin{bmatrix} 0.125 \\ 0.5 \end{bmatrix}u[k] \label{eqn:dbl_int_0.5}\\
        x[k+1]&=f_{1.5}(\cdot,\cdot)=\begin{bmatrix} 1 & 1.5 \\ 0 & 1 \end{bmatrix}x[k]+\begin{bmatrix} 1.125 \\ 1.5 \end{bmatrix}u[k] \nonumber
    \end{align}
The maximal control invariant set for each discretization, computed via Alg. 10.2 in \cite{Borrelli_MPC_book}, is plotted in Fig.~\ref{fig:dbl_int_no_downsample}.

We observe that $\mathcal{C}_{\infty,0.5} \subset \mathcal{C}_{\infty,1.5}$, suggesting there are states which can be safely controlled with a sampling time of $1.5$, but not a sampling time of $0.5$. 
This is an artifact of 
the fact that $\mathcal{C}_{\infty,1.5}$ does not consider constraint violation at the faster sampling time $T_s=0.5$. To show this, 
an optimal control problem was used to find a sequence of feasible inputs to drive $f_{1.5}$ from $x[0]$ to the origin (note that $x[0]\in \mathcal{C}_{\infty, 1.5},\text{ but }x[0]\not\in \mathcal{C}_{\infty, 0.5}$). 
Indeed, up-sampling and applying these inputs to $f_{0.5}$ at a higher frequency results in state constraint violation, as does applying these inputs to the continuous time model $f_c$, as shown in Fig.~\ref{fig:dbl_int_no_downsample}.

Since traditional control invariance for $f_{T_s}$ and $f_{MT_s}$ does not guarantee inter-sample constraint satisfaction, it cannot be readily used to reason about the safety of an adaptive-sampling control scheme. This is formalized in Remark~\ref{rmk:ctr_inv_bad}.

\begin{remark} \label{rmk:ctr_inv_bad}
    $\mathcal{C}_\infty$ for a coarse discretization of $f_c$ is not necessarily a subset of $\mathcal{C}_\infty$ for a finer discretization of $f_c$:
    \begin{equation*} 
       \big(M_j \geq M_i\big) \nRightarrow \big(\mathcal{C}_{\infty,M_jT_s}  \subseteq \mathcal{C}_{\infty,M_iT_s}\big) \label{eqn:bad_inc_a}
    \end{equation*}
\end{remark}

\begin{figure}
    \centering
    \includegraphics[width=1\linewidth]{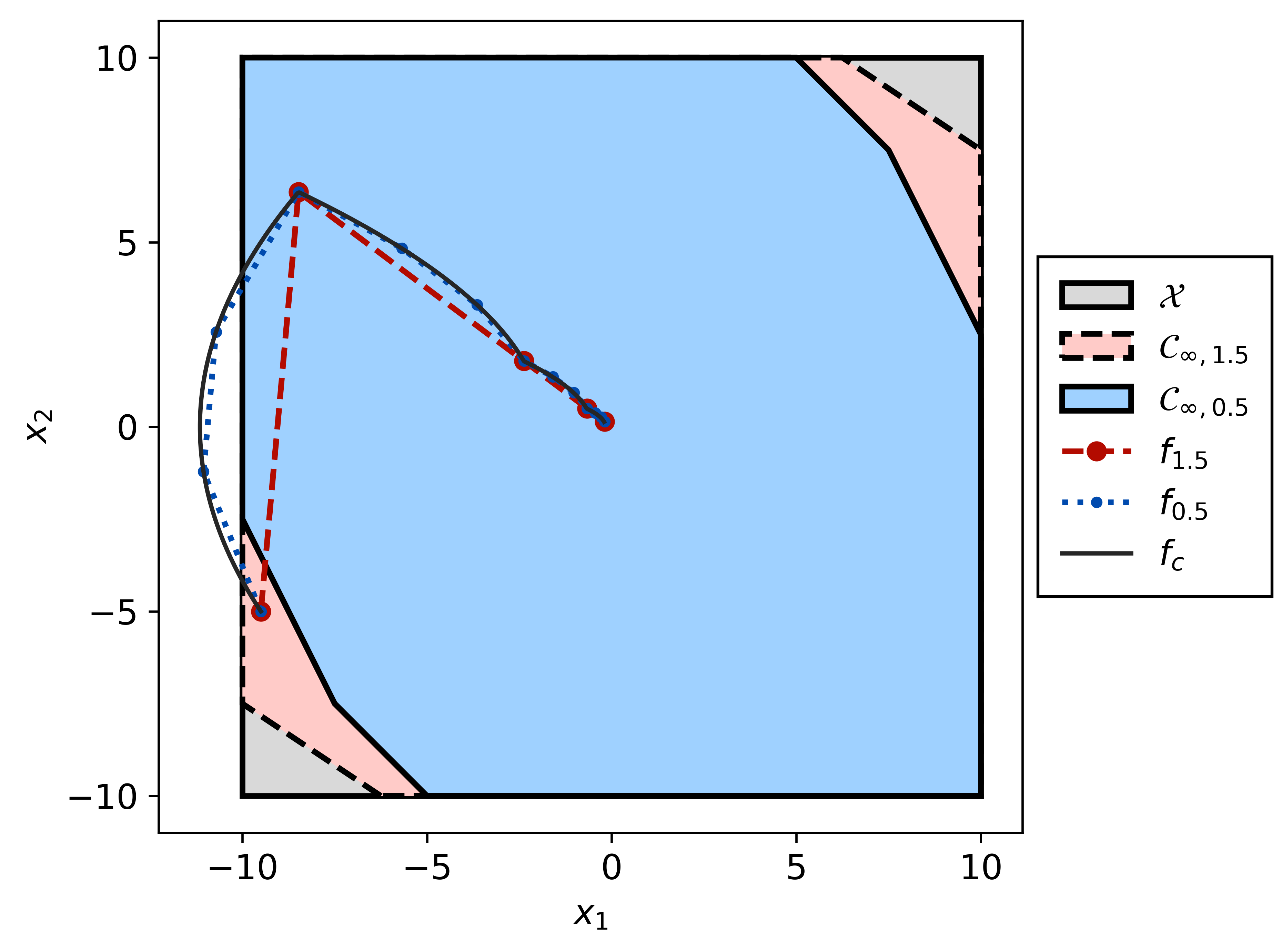}
    \caption{$\mathcal{C}_{\infty,T_s} \subset\mathcal{C}_{\infty,MT_s}$ for exact discretizations of a constrained double integrator with $T_s=0.5$ and $M=3$. Up-sampled optimal inputs calculated for $f_{MT_s}$ applied to $f_{T_s}$ and $f_c$ can violate constraints.}
    \label{fig:dbl_int_no_downsample}
\end{figure}
Motivated by Remark~\ref{rmk:ctr_inv_bad}, we introduce ``$M$-step hold control invariance," a generalization of control invariance that guarantees inter-sample constraint satisfaction, allowing us to properly study the effect of sampling time on constraint satisfaction. Nominal $M$-step hold control invariance guarantees constraint satisfaction at multiples of the fastest sampling time, while the robust version does so at all time instants.



\section{\emph{M}-Step Hold Control Invariance}\label{sec:mstep_ctr_inv}
A controller employs a Zero-Order Hold (ZOH) if it maintains a constant control input over each sampling interval 
\begin{align*}
    u(t) = u[k],~ \forall t \in \big[k\cdot T_s, (k+1)\cdot T_s\big).
\end{align*}
An ``$M$-step hold" extends this idea for multiple intervals. 
\begin{definition} A digital controller with sampling time $T_s$ is implemented with an \textit{\textbf{$M$-step hold}} if the input is allowed to change only every $M$ time samples: \label{def:mstep_hold}
    \begin{equation*}
        u[k]=u[k-1] \text{ if } \big(k>0 \text{ and mod} (k,M)\neq0\big), \  k\in\mathbb{N}_0 
    \end{equation*}
\end{definition}
We will refer to a controller satisfying Def.~\ref{def:mstep_hold} as an $M$-step hold controller. Note ZOH and $M$-step hold refer to control update rate properties, not discretization methods for generating $f^*_{d,T_s}$. A model-based M-step hold controller can calculate inputs using any $f^*_{d,T_s}$.
\begin{definition}
The set $\mathcal{C}^M_{f^*_{d,T_s}}\subseteq{\mathcal{X}}$ is \textit{\textbf{$M$-step hold control invariant}} for \eqref{eqn:dt_dyn} subject to \eqref{eqn:const_dt} if for all initial states in $\mathcal{C}^M_{f^*_{d,T_s}}$ there exists an $M$-step hold controller such that the system constraints are satisfied at all time samples:
    \begin{align*}
        & x[k]\in\mathcal{C}^M_{f^*_{d, T_s}} \Rightarrow \\
        & \exists \ u[k]\in\mathcal{U} ~:~ f^*_{d,T_s}\big(x[k],u[k]\big) \in\mathcal{C}^M_{f^*_{d, T_s}} \ \forall k\in\mathbb{N}_0\nonumber \\
        & \text{and } u[k]=u[k-1] \text{ if } \big(k>0 \text{ and mod} (k,M) \neq 0\big) \nonumber
    \end{align*}
\end{definition}

We denote with $\mathcal{C}^M_{\infty, f^*_{d,T_s}}$ the \textit{\textbf{maximal $M$-step hold control invariant set}}, i.e., the $M$-step hold control invariant set containing all other $\mathcal{C}^M_{f^*_{d,T_s}}\subseteq\mathcal{X}$. As in~\eqref{eqn:sim_not}, for a fixed discretization technique, we will use the 
simplified notation
\begin{gather}
    \mathcal{C}^M_{\infty, f^*_{d,T_s}}\rightarrow \mathcal{C}^M_{\infty,T_s}
    \label{eqn:msim_not}
\end{gather}
The set $\mathcal{C}^M_{\infty,f^*_{d,T_s}}$ is calculated using Alg.~\ref{alg:mcinf}, a modification of the standard fixed-point algorithm (see Alg. 10.2 in \cite{Borrelli_MPC_book}). 

\begin{definition} The \textit{\textbf{$M$-step hold precursor set}} $Pre^M(\mathcal{S})$ of target set $\mathcal{S}$ for \eqref{eqn:dt_dyn} is the set of all states $x[0]\in\mathbb{R}^n$ for which there exists a constant input $u\in\mathcal{U}$ such that all of the next $M$ states are in $\mathcal{S}$: \label{def:mpre}
    \begin{align}
        & Pre^M(\mathcal{S}) = \big\{x[0]\in\mathbb{R}^n:\exists \ u\in\mathcal{U} \text{ s.t. }  \label{eqn:mpre} \\
        & ~~~~~x[k+1]=f^*_{T_s}\big(x[k],u\big)\in\mathcal{S}, \ \forall k\in\{0,\dots,M-1\}\big\} \nonumber
    \end{align}
\end{definition}


The formula for computing $Pre^M(\set{S})$ for discrete LTI models is in Sec.~\ref{sec:computation}. By Def.~\ref{def:mpre}, $Pre^1(\set{S})$ is equivalent to the traditional precursor set in \cite{Borrelli_MPC_book}. Thus,
$\mathcal{C}_{f^*_{d,T_s}} = \mathcal{C}^1_{f^*_{d,T_s}}$ and $\mathcal{C}_{\infty,f^*_{d,T_s}}= \set{C}^1_{\infty,f^*_{d,T_s}}$. Any traditional control invariant set can be generated using $M$-step hold control invariance by selecting $M=1$. However, any $M$-step hold control invariant set cannot be generated using traditional control invariance. Thus, $M$-step control invariance is a generalization of traditional control invariance.

\begin{algorithm}[t]
    \caption{Computation of $\mathcal{C}^M_\infty$} \label{alg:mcinf}
    \begin{algorithmic}
        \Require $f^*_{d,T_s}, \mathcal{X}, \mathcal{U}, M$
        \Ensure $\mathcal{C}^M_\infty$
        \State $\Omega^M_0 \gets \mathcal{X}, i \gets -1$
        \Repeat 
            \State $i \gets i+1$
            \State $\Omega^M_{i+1}\gets Pre^M(\Omega^M_i)\cap\Omega^M_i$
        \Until $\Omega^M_{i+1}=\Omega^M_i$
        \State $\mathcal{C}^M_\infty \gets \Omega^M_{i+1}$
    \end{algorithmic}
\end{algorithm}


$M$-step hold control invariant sets guarantee inter-sample constraint satisfaction, allowing us to properly study the effect of control sampling time. This is formalized in Thm.~\ref{thm:mset_inclusion}. 
\begin{lemma} \label{lem:mpre_inclusion}
For fixed $\mathcal{S}$, $Pre^{M+1}(\mathcal{S})\subseteq Pre^M(\mathcal{S}).$
\end{lemma} 
\begin{proof}
    By Def.~\ref{def:mpre},
    \begin{align}
    Pre^{M+1}(\mathcal{S})&  = \big\{x[0]\in\mathbb{R}^n:\exists \ u\in\mathcal{U}  \nonumber \\
    & ~~~~~\text{s.t. } x[k+1]=f^*_{T_s}\big(x[k],u\big)\in\mathcal{S}, \nonumber \\ 
    & ~~~~~ k\in\{0,\dots,M\}\big\} \nonumber \\
    & = 
 \big\{x[0]\in\mathbb{R}^n:\exists \ u\in\mathcal{U} \nonumber \\
 & ~~~~~\text{s.t. } x[k+1]=f^*_{T_s} \big(x[k],u\big)\in\mathcal{S}, \nonumber \\
     & ~~~~~ k\in\{0,\dots,M-1\}\big\} \nonumber \\
     &~~~\cap \big\{x[0]\in\mathbb{R}^n: ~x[k+1]=f^*_{d,T_s}\big(x[k],u\big), \nonumber \\
    & ~~~~~k\in\{0,\dots,M\}, \  x[M+1]\in{S} \big\}\nonumber\\
     &=Pre^M(\mathcal{S}) \label{eqn:mcall} \\
     & ~~~\cap \big\{x[0]\in\mathbb{R}^n:~ x[k+1]=f^*_{d,T_s}\big(x[k],u\big), \nonumber \\
    & ~~~~~k\in\{0,\dots,M\}, \  x[M+1]\in{S}\big\} \nonumber
    \end{align}
    Note that for any sets $\mathcal{D}$, $\mathcal{E}$, $\mathcal{F}$: $(\mathcal{D}=\mathcal{E} \cap \mathcal{F}) \Rightarrow (\mathcal{D} \subseteq\mathcal{E})$
    which completes the proof:
    \begin{equation*}
        Pre^{M+1}(\mathcal{S})\subseteq Pre^M(\mathcal{S})
    \end{equation*}
\end{proof}

\begin{theorem}\label{thm:mset_inclusion}
    Given a discrete model $f_{T_s}$, if Alg.~\ref{alg:mcinf} converges in finite time for a given $M$ and $M+1$, then    the maximal $(M+1)$-step hold control invariant set is a subset of the maximal $M$-step hold control invariant set:
    \begin{equation*}
        \mathcal{C}^{M+1}_{\infty, T_s}\subseteq \mathcal{C}^{M}_{\infty, T_s} 
    \end{equation*}
\end{theorem}

\proof
    We show this using induction.
    Initialize two instances of Alg.~\ref{alg:mcinf} for $f_{T_s}$, computing $\mathcal{C}^{M+1}_{\infty,T_s}$ and $\mathcal{C}^{M}_{\infty,T_s}$ respectively. 
    
    Consider $i=0$, where $\Omega^{M+1}_0 = \Omega^M_0 = \mathcal{X}$.  
   By Lem.~\ref{lem:mpre_inclusion}:
    \begin{equation*}
        Pre^{M+1}(\Omega^{M+1}_0)\subseteq Pre^M(\Omega^M_0),
    \end{equation*}
    from which it follows that 
    \begin{subequations}
    \begin{align*}
        \big(Pre^{M+1}(\Omega^{M+1}_0)\cap \Omega^{M+1}_0\big) &\subseteq \big(Pre^M(\Omega^M_0)\cap\Omega^M_0\big)  
    \end{align*}
    \end{subequations}
and therefore $\Omega^{M+1}_1 \subseteq\Omega ^M_1.$

Now consider iteration $i$ of Alg.~\ref{alg:mcinf}, and assume $\Omega^{M+1}_i \subseteq \Omega^M_i$. Invoking Lem.~\ref{lem:mpre_inclusion} on the set $\Omega^{M+1}_i$ gives:
\begin{subequations}
    \begin{align}
    Pre^{M+1}(\Omega^{M+1}_i)& \subseteq Pre^M(\Omega^{M+1}_i) \nonumber \\
    & \subseteq Pre^M\big(\Omega^{M+1}_i\cup(\Omega^M_i\setminus\Omega^{M+1}_i) \big) \label{eq:thm1ref1}\\
    & = Pre^M(\Omega^{M}_i), \label{eq:thm1ref2}
\end{align}
\end{subequations}
where \eqref{eq:thm1ref1} follows since including additional states in the argument of $Pre^M(\cdot)$ cannot make the output smaller,
and \eqref{eq:thm1ref2} follows from $\Omega^M_i=\Omega^{M+1}_i\cup (\Omega^M_i\setminus\Omega^{M+1}_i). $

Thus we have shown for Alg.~\ref{alg:mcinf} that \textit{1)} at iteration $i=0$, $\Omega_0^{M+1} \subseteq \Omega_0^M$, and \textit{2)} if $\Omega_i^{M+1} \subseteq \Omega_i^M$, then $\Omega_{i+1}^{M+1} \subseteq \Omega_{i+1}^M$. We conclude by induction that $\Omega_{\infty}^{M+1} \subseteq \Omega_{\infty}^M$, and therefore it follows by design of Alg.~\ref{alg:mcinf} that
\begin{align*}
\mathcal{C}^{M+1}_{\infty,T_s}&\subseteq\mathcal{C}^{M}_{\infty, T_s}
\end{align*}

\endproof

Consider the earlier example \eqref{eqn:dbl_int_0.5} of a discretized constrained double integrator with $T_s=0.5$. Figure \ref{fig:mset_inclusion} demonstrates the intuitive evolution of $\set{C}^M_{\infty,T_s}$ with $M$: sets with larger $M$ are subsets of those with smaller $M$. 
\begin{equation*}
    (M_j\geq M_i) \Rightarrow (\mathcal{C}^{M_j}_{\infty,T_s}\subseteq \mathcal{C}^{M_i}_{\infty,T_s}) \label{eqn:good_inc}
\end{equation*}
\begin{figure}
    \centering
    \includegraphics[width=1\linewidth]{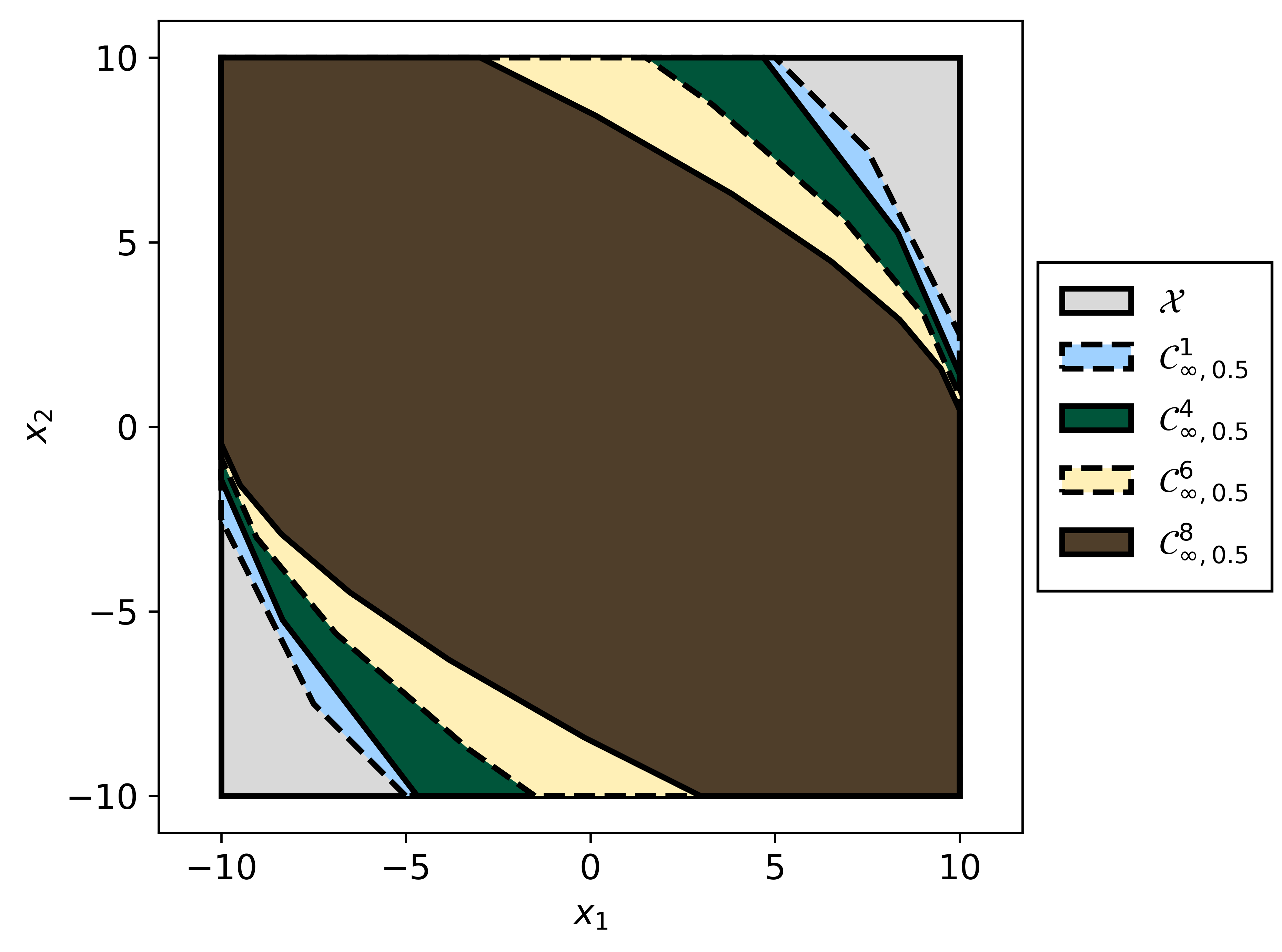}
    \caption{
    $\mathcal{C}^M_{\infty,T_s}$ with larger $M$ are subsets of those with smaller $M$.}
    \label{fig:mset_inclusion}
\end{figure}

\begin{figure}
    \centering
    \includegraphics[width=1\linewidth]{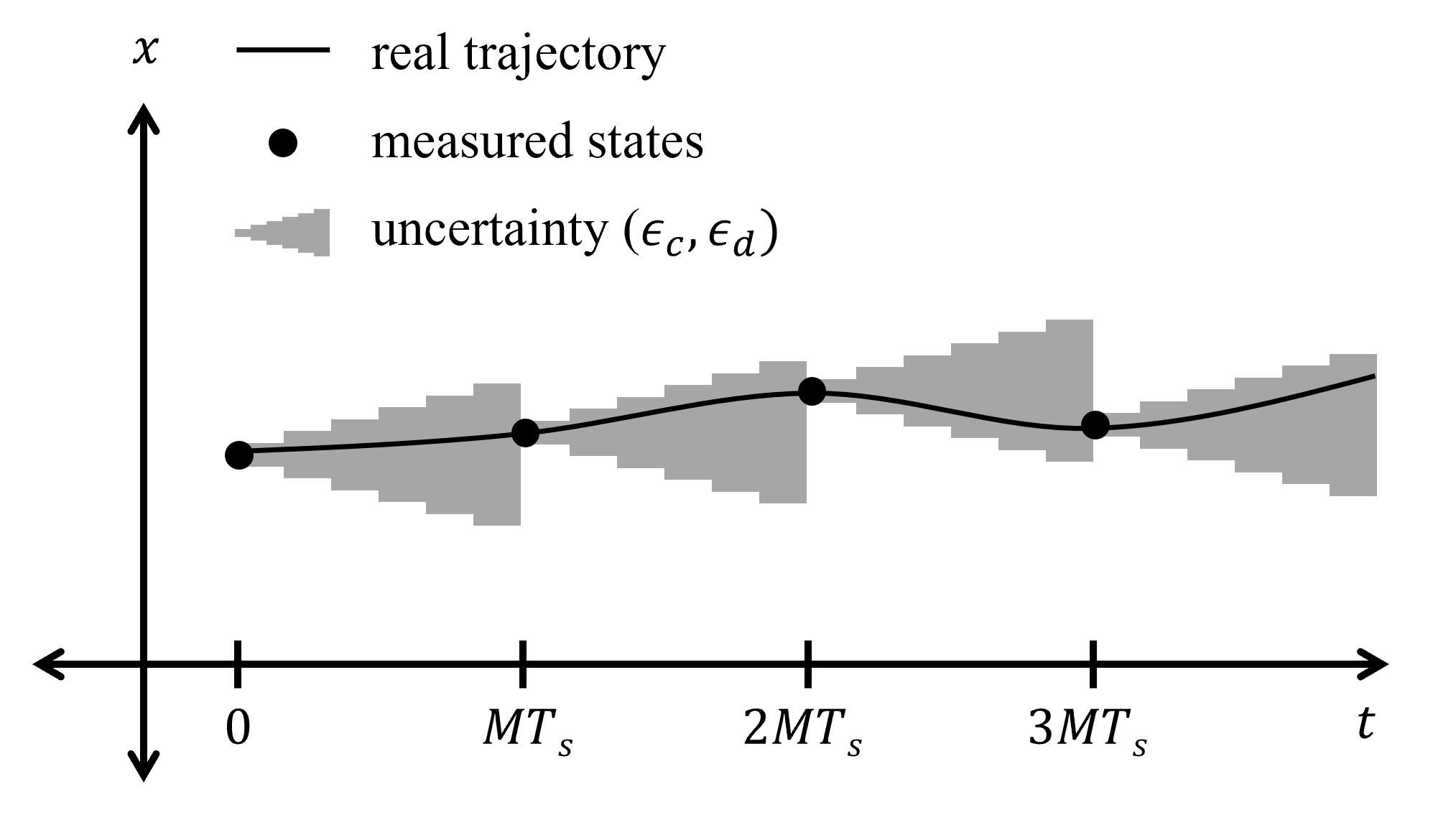}
    \caption{Visualization of $\set{W}[\text{mod}(k,M)]$ for a 1D system. The uncertainty ($\epsilon_c$ and $\epsilon_d$) resets with each measurement.}
    \label{fig:w_reset}
\end{figure}

The behavior seen in Fig.~\ref{fig:mset_inclusion}, guaranteed by Thm.~\ref{thm:mset_inclusion}, makes $M$-step hold control invariance a suitable tool for designing adaptive-sampling controllers. 
We emphasize two key points. 
First, the polytopes show the state-space regions for which the system can safely run in open loop under an appropriate control law during the $M$-step hold. 
Second, Fig.~\ref{fig:mset_inclusion} shows how to safely switch between values of $M$, i.e. switch between sampling times. For any $M_2>M_1$, a controller using a fixed model $f^*_{d_,T_s}$ can switch from an $M_2$-step hold controller to an $M_1$-step hold controller whenever  $x[k]\in\set{C}^{M_2}_{\infty,T_s}$. Continued constraint satisfaction is guaranteed because all $x[k]$ in $\set{C}^{M_2}_{\infty,T_s}$ are also in $\set{C}^{M_1}_{\infty,T_s}$. 
To safely switch from an $M_1$-step hold controller to an $M_2$-step hold controller, the state must first be driven back into $\mathcal{C}^{M_2}_{\infty,T_s}$. 
This idea will be leveraged in future work to design controllers which modify sampling time based on task complexity.



\section{Robustification}\label{sec:robust}

Section~\ref{sec:mstep_ctr_inv} showed that $M$-step hold control invariant sets 
for a discrete-time model $f^{\star}_{d, T_s}$ 
can be used to guarantee constraint satisfaction at every sampling time $T_s$ despite using an $M$-step hold controller. 
But while $f^{\star}_{d, T_s}$ is used to calculate the optimal control inputs, the inputs are applied to a continuous real-world system, resulting in different closed-loop behavior than captured by $f^{\star}_{d, T_s}$ and contained in the corresponding $M$-step hold control invariant set.

Consider the unknown ``real" continuous time system 
    \begin{align}
    &\dot{x}_r(t)=f_{r}(x_r(t),u_r(t)), \label{eqn:real_dyn}\\
    &{x}_r(t) \in \mathcal{X},~ u_r(t) \in \mathcal{U}
\end{align}
whose behavior is approximated by the model $f_c$ \eqref{eqn:ct_dyn}. 
Additional approximation errors are introduced when the model $f_c$ is discretized. 
At time step $k$, the future state $x[k+T|k]$ at time $k+T$ predicted using a discretized model of $f_c$ will have an error\footnote{In \eqref{eq:error_bounding_one}, we consider that $x[k|k]=x_r(kT_s)$ and that the sampled input applied to the discrete time prediction model is the same as applied to the the true continuous time model, held constant between samples.} compared to the true future state of the real-world system $x_r\big((k+T)T_s\big)$, i.e.
\begin{align}\label{eq:error_bounding_one}
    \big\|x[k+T|k] - x_r\big((k+T)\cdot T_s\big)\big\|_2 = \big\|\epsilon[T|k]\big\|_2,
\end{align}
where $\epsilon[T|t]$ is the sum of the modeling error $\epsilon_c$, introduced by inaccuracies in $f_c$ such as simplified dynamics and parameter uncertainty, and the discretization error $\epsilon_d$, for example as described in \cite{STEIN20112626}:
\begin{equation}\label{eqn:errors}
    \epsilon[T|k]=\epsilon_c[T|k]+\epsilon_d[T|k]. 
\end{equation}

Extensive literature has focused on bounding these modeling errors \cite{tomlin2017, hur2019, VOELKER2013943} and discretization errors \cite{roy2010}; a review on this topic is beyond this paper's scope.
Here we assume the modeling and discretization error $\epsilon[T|k]$ grows along the prediction horizon $T$, and resets whenever a new state measurement becomes available, as depicted in Fig.~\ref{fig:w_reset}. 


In order to extend the methods from Sec.~\ref{sec:mstep_ctr_inv} to real-world systems in closed-loop with an $M$-step hold controller, we therefore consider a discretized model $f^*_{d,T_s}$ with an added disturbance $w$ with time-varying bound:
\begin{subequations}
    \begin{align}
        x[k+1]& =f^*_{d,T_s}\big(x[k],u[k],w[k]\big) \label{eqn:rob_dt_dyn}\\
    x[k]\in\set{X}&,~ u[k]\in\set{U},~ w[k]\in\set{W}\big[\text{mod}(k,M)\big] \label{eqn:rob_const_dt}
    \end{align}
\end{subequations}
where $\mathcal{W}[k]$ captures the time-varying combined effects of $\epsilon_c$ and $\epsilon_d$.
$\mathcal{W}\big[\text{mod}(k,M)\big]$
captures the fact that uncertainty bounds reset at each sampling time (assuming perfect state measurement), shown in Fig.~\ref{fig:w_reset}. 
Robustifying $M$-step control invariance against $\mathcal{W}\big[\text{mod}(k,M)\big]$ will guarantee constraint satisfaction of the true system $f_r$ \eqref{eqn:real_dyn} for all time. 

\begin{definition}\label{def:rob_ctr_inv_dt} $\mathcal{RC}^M_{f^*_{d, T_s}}$ is a \textit{\textbf{robust $M$-step hold control invariant set}} for \eqref{eqn:rob_dt_dyn} subject to \eqref{eqn:rob_const_dt} if for all initial states in $\mathcal{RC}_{f^*_{d,T_s}}$ there exists a controller with an $M$-step hold such that constraints are robustly satisfied at all time samples: 
    \begin{align}
    & x[k]\in \mathcal{RC}_{f^*_{d, T_s}} \Rightarrow \exists \ u[k]\in\mathcal{U} \ \text{s.t.} \label{eqn:rob_ctr_inv_dt} \\
    &  ~~~\ x[k+1]=f^*_{d,T_s}\big(x[k],u[k],w[k]\big) \in \mathcal{RC}_{f^*_{d, T_s}} \nonumber \\
& ~~~\ \forall w[k]\in\set{W}\big[\text{mod}(k,M)\big], ~\forall k\in\mathbb{N}_0 \nonumber\\
    & \text{and } u[k]=u[k-1] \text{ if } \big(k>0 \text{ and mod} (k,M)\neq0\big) \nonumber 
    \end{align}
\end{definition}

The \textbf{\textit{maximal robust $M$-step hold control invariant set}} $\mathcal{RC}^M_{\infty,f^*_{d, T_s}}$ is the largest robust $M$-step hold control invariant set containing all other $\mathcal{RC}^M_{f^*_{d, T_s}}\in\set{X}$. 
It is computed using Def.~\ref{def:robmpre} and Alg.~\ref{alg:rob_mcinf}, robust analogues of Def.~\ref{def:mpre} and Alg.~\ref{alg:mcinf}. 

\begin{definition} \label{def:robmpre} The \textit{\textbf{robust $M$-step hold precursor set}} $Pre^M(\mathcal{S},\mathcal{W})$ of target set $\mathcal{S}$ for \eqref{eqn:rob_dt_dyn} is the set of all states $x[0]\in\mathbb{R}^n$ for which there exists a held input $u\in\mathcal{U}$ such that the next $M$ states are robustly in $\mathcal{S}$.
\begin{align}
    & Pre^M(\mathcal{S,W}) = \big\{x[0]\in\mathbb{R}^n:\exists \ u\in\mathcal{U} \text{ s.t. } \nonumber \\
    & ~~~~~~x[k+1]=f^*_{T_s}\big(x[k],u,w[k]\big)\in\mathcal{S}, \nonumber \\ 
    & ~~~~\forall k\in\{0,\dots,M-1\}, \forall w[k]\in\set{W}\big[\text{mod}(k,M)\big]\big\} \label{eqn:rob_mpre}
\end{align}
\end{definition}
\noindent The formula for computing $Pre^M(\set{S},\set{W})$ for discrete LTI models is in Section \ref{sec:computation}.

\begin{algorithm}[t]
    \caption{Computation of  $\mathcal{RC}^M_\infty$} \label{alg:rob_mcinf}
    \begin{algorithmic}
        \Require $f^*_{d,T_s}, \mathcal{X}, \mathcal{U}, \set{W}, M$
        \Ensure $\mathcal{RC}^M_\infty$
        \State $\Omega^M_0 \gets \mathcal{X}, i \gets -1$
        \Repeat 
            \State $i \gets i+1$
            \State $\Omega^M_{i+1}\gets Pre^M(\Omega^M_i,\set{W})\cap\Omega^M_i$
        \Until $\Omega^M_{i+1}=\Omega^M_i$
        \State $\mathcal{RC}^M_\infty \gets \Omega^M_{i+1}$
    \end{algorithmic}
\end{algorithm}

Importantly, robustifying $\set{C}^M_{\infty,T_s}$ against $\mathcal{W}$ does not change its evolution with respect to $M$, as stated in Thm.~\ref{thm:rob_set_inclusion}. 

\begin{lemma}\label{lem:robustmpre_inclusion}
    For fixed $\mathcal{S}$, $Pre^{M+1}(\mathcal{S}, \mathcal{W})\subseteq Pre^M(\mathcal{S}, \mathcal{W}).$
\end{lemma}
\proof
    By Def.~\ref{def:robmpre},
    \begin{align}
    Pre^{M+1}(\mathcal{S},\mathcal{W})&  = \big\{x[0]\in\mathbb{R}^n:\exists \ u\in\mathcal{U}  \nonumber \\
    & ~~~\text{s.t. } x[k+1]=f^*_{T_s}\big(x[k],u,w[k]\big)\in\mathcal{S}, \nonumber \\ 
    & ~~~ \forall w[k] \in \mathcal{W}\big[\text{mod}(k,M+1)\big], \nonumber\\
    & ~~~ k\in\{0,\dots,M\}\big\} \nonumber \\
    = & ~  
 \big\{x[0]\in\mathbb{R}^n:\exists \ u\in\mathcal{U} \nonumber \\
 & ~~ ~\text{s.t. } x[k+1]=f^*_{T_s} \big(x[k],u,w[k]\big)\in\mathcal{S}, \nonumber \\
     & ~~~ \forall w[k] \in \mathcal{W}\big[\text{mod}(k,M+1)\big], \nonumber\\
     & ~~~ k\in\{0,\dots,M-1\}\big\} \nonumber \\
     &~\cap \nonumber\\ 
    &~\big\{x[0]\in\mathbb{R}^n: \nonumber\\
    & ~~~x[k+1]=f^*_{d,T_s}\big(x[k],u,w[k]\big), \nonumber \\
    & ~~~k\in\{1,\dots,M\}, \  x[M+1]\in{S} \nonumber\\
    & ~~~\forall w[k] \in \mathcal{W}\big[\text{mod}(k,M+1)\big]\big\} \nonumber\\
     = &~Pre^M(\mathcal{S}, \mathcal{W}) \label{eqn:robustmcall} \\
     & ~\cap  \nonumber \\ 
     &~\big\{x[0]\in\mathbb{R}^n: \nonumber\\
    & ~~~x[k+1]=f^*_{d,T_s}\big(x[k],u,w[k]\big), \nonumber \\
    & ~~~k\in\{1,\dots,M\}, \  x[M+1]\in{S} \nonumber\\
    & ~~~\forall w[k] \in \mathcal{W}\big[\text{mod}(k,M+1)\big]\big\}. \nonumber
    \end{align}
    Note that for any sets $\mathcal{D}$, $\mathcal{E}$, $\mathcal{F}$ that
    \begin{equation}
        (\mathcal{D}=\mathcal{E} \cap \mathcal{F}) \Rightarrow (\mathcal{D} \subseteq\mathcal{E}).\label{eqn:rob_set_int_inc2}
    \end{equation}
    Applying \eqref{eqn:rob_set_int_inc2} to \eqref{eqn:robustmcall} completes the proof: 
    \begin{equation*}
        Pre^{M+1}(\mathcal{S}, \mathcal{W})\subseteq Pre^M(\mathcal{S},\mathcal{W})
    \end{equation*}
\endproof 

\begin{theorem}\label{thm:rob_set_inclusion}
    Given a discrete model $f_{T_s}$, if Alg.~\ref{alg:rob_mcinf} converges in finite time for given $M$ and $M+1$, then the robust maximal $(M+1)$-step hold control invariant set is a subset of the robust maximal $M$-step hold control invariant set:
    \begin{align*}
        &\mathcal{RC}^{M+1}_{\infty, T_s}\subseteq \mathcal{RC}^{M}_{\infty, T_s} ~~ \forall M\in\mathbb{N}_+ \label{eqn:rob_set_inc} 
    \end{align*}
\end{theorem}
\begin{proof}
We show this using induction.
Initialize two instances of Alg.~\ref{alg:rob_mcinf} for $f_{T_s}$, computing  $\mathcal{RC}^{M+1}_{\infty,T_s}$ and  $\mathcal{RC}^{M}_{\infty,T_s}$. 
    
Consider $i=0$, where $\Omega^{M+1}_0 = \Omega^M_0 = \mathcal{X}$.
By Lem.~\ref{lem:robustmpre_inclusion}:
    \begin{equation*}
        Pre^{M+1}(\Omega^{M+1}_0, \mathcal{W})\subseteq Pre^M(\Omega^M_0, \mathcal{W}),\label{eqn:alg1_init}
    \end{equation*}
    from which it follows that
    \begin{align*}
        \big(Pre^{M+1}(\Omega^{M+1}_0, \mathcal{W})\cap \mathcal{X}\big) &\subseteq \big(Pre^M(\Omega^M_0, \mathcal{W})\cap \mathcal{X}\big)
    \end{align*}
    and therefore $\Omega^{M+1}_1 \subseteq\Omega ^M_1.$

Now consider iteration $i$ of Alg.~\ref{alg:rob_mcinf}, and assume $\Omega^{M+1}_i \subseteq \Omega^M_i$.
Invoking Lem.~\ref{lem:robustmpre_inclusion} on the set $\Omega^{M+1}_i$ gives:
        \begin{align}
        &Pre^{M+1}(\Omega^{M+1}_i,\mathcal{W})\subseteq \nonumber\\
        &~~~~~~~~~ \subseteq Pre^M(\Omega^{M+1}_i,\mathcal{W}) \nonumber\\ 
        &~~~~~~~~~ \subseteq Pre^M\big(\Omega_i^{M+1} \cup (\Omega_i^M \setminus \Omega_i^{M+1} ), \mathcal{W} \big)  \nonumber\\
        & ~~~~~~~~~= Pre^M(\Omega_i^M, \mathcal{W})\label{eqn:use_lem_2}, 
    \end{align}
where \eqref{eqn:use_lem_2} follows from $\Omega^M_i=\Omega^{M+1}_i\cup (\Omega^M_i\setminus\Omega^{M+1}_i).$


    Thus we have shown for Alg.~\ref{alg:rob_mcinf} that \textit{i)} at $i=0$, 
    $\Omega^{M+1}_{1} \subseteq \Omega^{M}_1$, and 
  \textit{ii)} if $\Omega^{M+1}_{i} \subseteq \Omega^{M}_i$, then $\Omega^{M+1}_{i+1} \subseteq \Omega^{M}_{i+1}$. 
    We conclude by induction that $\Omega^{M+1}_\infty \subseteq \Omega^M_\infty$, and therefore it follows by design of Alg.~\ref{alg:rob_mcinf} that
    \[ \text{ }\mathcal{RC}^{M+1}_{\infty,T_s}\subseteq \text{ }\mathcal{RC}^{M}_{\infty, T_s}.\]
\end{proof}

If at time $t$ the ``real" system state \eqref{eqn:real_dyn} is in $\mathcal{RC}^{M}_{\infty, Ts}$, we can use the discretized model \eqref{eqn:rob_dt_dyn} and a control frequency $MT_s$ to control the system in $\mathcal{RC}^{M}_{\infty, Ts}$ for all time steps with an appropriate controller. This is stated in Thm.~\ref{thm:control}. 

\begin{assumption}\label{assm:noise}
    The error \eqref{eqn:errors} between the real, continuous time system $f_r$ \eqref{eqn:real_dyn} 
    and a discrete-time system approximation $f_{d,T_s}$ \eqref{eqn:dt_dyn} 
    of a nominal model $f_c$ 
    \eqref{eqn:ct_dyn} 
    are bounded by a known, time-varying set $\mathcal{W}(\cdot)$, for all admissible inputs $u(t)\in\mathcal{U}$ so that
    \begin{align*}
        &~~~~~~x_r(\tilde{t}) \in x[k] \oplus \mathcal{W}\big[\text{mod}(k,M)\big]\\
        & \tilde{t} \in \big[kT_s, (k+1)T_s\big),~ \forall k \geq 0,~\forall u[k]\in\mathcal{U}
    \end{align*}
    where $\text{mod}(k,M)$ counts the time steps since the last state measurement. 
\end{assumption}
We refer to \cite{STEIN20112626, tomlin2017, hur2019, VOELKER2013943, roy2010} for examples of calculating uncertainty sets $\mathcal{W}$. If the estimated sets $\mathcal{W}$ are very conservative, the robust $M$-step hold invariant set \eqref{def:rob_ctr_inv_dt} may be small or even empty. 

\begin{theorem}\label{thm:control}
    Consider an unknown  real, continuous time system \eqref{eqn:real_dyn} and a corresponding discrete-time system $f_{d,T_s}$ and disturbance set $\mathcal{W}$ such that Assm.~\ref{assm:noise} holds. 
    Assume \eqref{eqn:rob_dt_dyn} and $\mathcal{W}$ are used to construct a Robust $\mathcal{C}^{M}_{\infty, Ts}$ \eqref{eqn:rob_ctr_inv_dt}.
    If $x_r(0) \in \mathcal{RC}^{M}_{\infty, Ts}$ and $ x[0]=x_r(0) $, then
    there exists an $M$-step hold feedback controller $\pi(\cdot)$
    such that 
    the real system \eqref{eqn:real_dyn} in closed-loop with $\pi(\cdot)$ with update frequency $MT_s$ will guarantee that $$x_r(t) \in \mathcal{RC}^{M}_{\infty, Ts}~ \forall t \geq 0.$$ 
\end{theorem}
\begin{proof}
Consider the system \eqref{eqn:rob_dt_dyn}. By definition of $\mathcal{RC}^{M}_{\infty, Ts}$ \eqref{eqn:rob_ctr_inv_dt}, if $x[k] \in \mathcal{RC}^{M}_{\infty, Ts}$ then there exists an $M$-step hold input $\pi(x_k) = u[k] \in \mathcal{U}$ such that $x[k+1] \in \mathcal{RC}^{M}_{\infty, Ts}$. 
By Assm.~\ref{assm:noise}, the underlying real-world system \eqref{eqn:real_dyn} trajectory is captured within the noise bound $\mathcal{W}$ used to construct $\mathcal{RC}^{M}_{\infty, Ts}$, and so $x_r(t) \in \mathcal{RC}^{M}_{\infty, Ts}~ \forall t \geq 0.$ 
\end{proof}



We consider
the earlier example of a constrained system discretized with $T_s=0.5$ \eqref{eqn:dbl_int_0.5}, now with additive uncertainty:
\begin{equation*}
    x[k+1]=f_{0.5}(\cdot, \cdot, \cdot)=\begin{bmatrix} 1 & 0.5 \\ 0 & 1 \end{bmatrix}x[k]+\begin{bmatrix} 0.125 \\ 0.5 \end{bmatrix}u[k] + w[k]
\end{equation*}
\begin{equation*}
    \big|w[k]\big|\leq \big(\text{mod}(k,M)+1\big)
    \begin{bmatrix} 
        0.2 \\ 0.2
    \end{bmatrix}\label{eqn:wset_lti}
\end{equation*}
 $\set{RC}^M_{\infty,T_s}$ sets for $M\in\{1,4,6,8\}$ are shown in Fig.~\ref{fig:rob_mcinf}; as guaranteed by Thm.~\ref{thm:rob_set_inclusion}, $\set{RC}^M_{\infty,T_s}$ sets for larger $M$ are subset of those for smaller $M$.
Thus an $M$-step hold controller can switch between $M_2$ and $M_1$ ($M_2>M_1$) at any state within  $\set{RC}^{M_2}_{\infty,T_s}$, with recursive feasibility guaranteed by Thm.~\ref{thm:control}.
\begin{figure}
    \centering
    \includegraphics[width=1\linewidth]{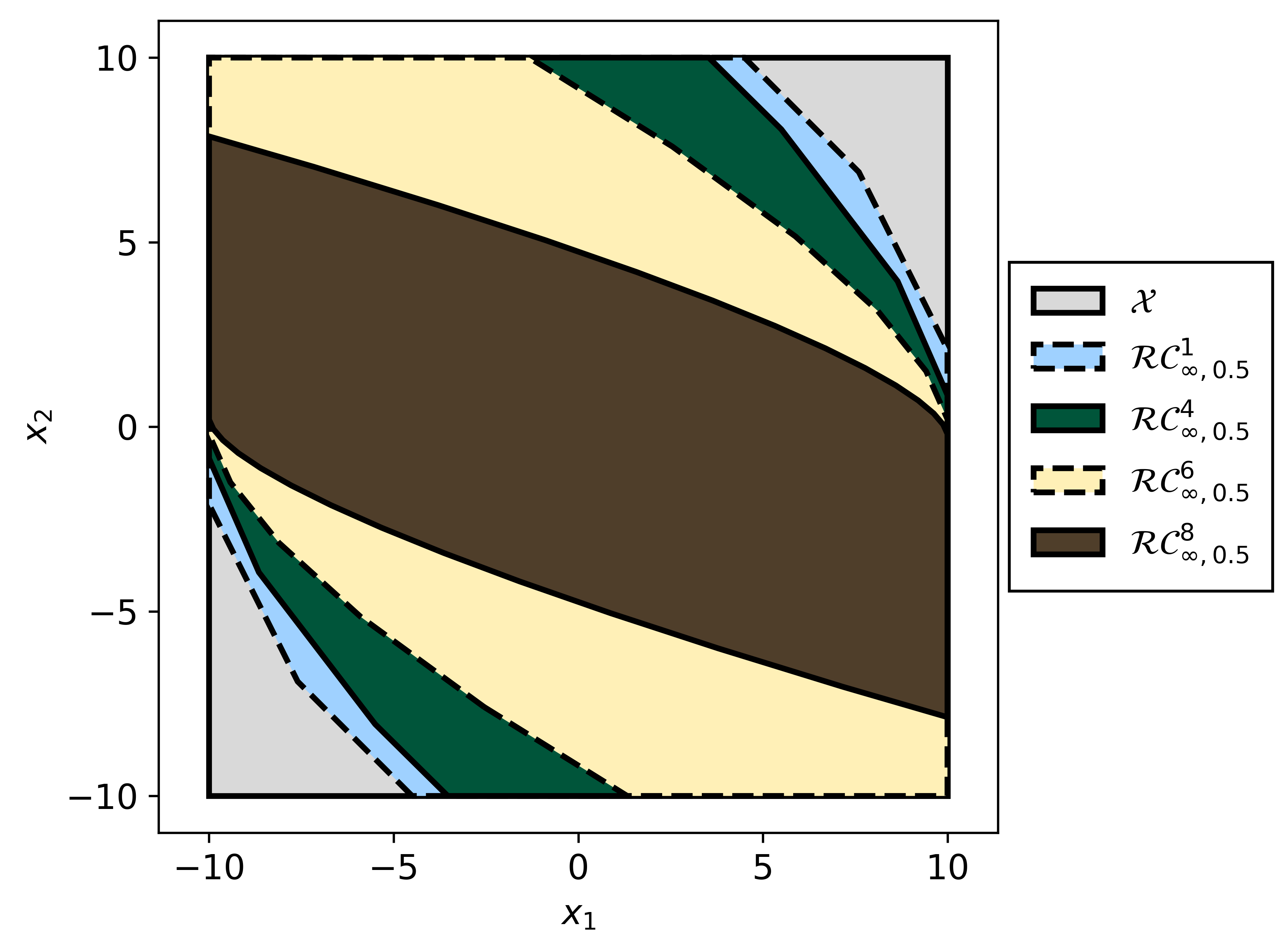}
    \caption{
     $\set{RC}^M_{\infty,T_s}$ with larger $M$ are subsets of those with smaller $M$.}
    \label{fig:rob_mcinf}
\end{figure}

Robust $M$-step hold control invariance paves the way for safety-critical adaptive-sampling control. Future work will design optimal $M$-step hold controllers with adaptive $M$, 
guaranteeing recursive feasibility through $\set{RC}^M_{\infty,T_s}$ sets. 

\section{Computation for Discrete LTI Models}\label{sec:computation} 
The definitions and algorithms of (robust) $M$-step hold control invariance make no assumptions of linearity. Many methods exist for calculating controllable sets for nonlinear models \cite{tomlin2018, Xiang2019, johnson2018, althoff2024}. For discrete LTI models, $\set{C}^M_{\infty,T_s}$ and $\set{RC}^M_{\infty,T_s}$ are easy to compute because $Pre^M(\set{S})$ and $Pre^M(\set{S},\set{W})$ become simple polytope projections. 

Consider a discretized model with polytopic constraints:
\begin{subequations}
\begin{gather}
    x[k+1]=A^*_{T_s}x[k]+B^*_{T_s}u[k] \label{eqn:dt_lti}\\
    \mathcal{S}=\{x : Hx \leq h\},~
    \mathcal{U} = \{u:H_u u\leq h_u\} \label{eqn:h_rep}
\end{gather}
\end{subequations}
Substituting \eqref{eqn:dt_lti},\eqref{eqn:h_rep} into Def.~\ref{def:mpre} generates:
\begin{align}
    Pre^M(\mathcal{S}) = 
    \left\{x\in\mathbb{R}^n: \exists ~u\in\mathbb{R}^m ~ s.t. ~\hat{H} 
    \begin{pmatrix} x \\ u \end{pmatrix} \leq \hat{h}\right\} \nonumber\\
    \hat{H} = \begin{bmatrix} 
    HA_{T_s} & HB_{T_s} \\ 
    HA^2_{T_s} & H(A_{T_s}B_{T_s}+B_{T_s})  \\
    \vdots & \vdots \\
    HA_{T_s}^M & H\left(\sum\limits_{n=1}^{M}{A_{T_s}^{M-n}B_{T_s}}\right)\\
    0 & H_u
    \end{bmatrix} \quad
     \hat{h}=\begin{bmatrix} h \\ h\\ \vdots \\ h \\ h_u \end{bmatrix} \label{eqn:mpre_lti_b}
\end{align}
Using \eqref{eqn:mpre_lti_b} in Alg. \ref{alg:mcinf} is sufficient to calculate $\set{C}^M_{\infty,T_s}$. 

Similarly, to compute $\set{RC}^M_{\infty,T_s}$, we first augment the nominal model with  additive uncertainty:
\begin{subequations}
\begin{gather}
    x[k+1]=A_{T_s}x[k]+B_{T_s}u[k]+w[k] \label{eqn:rob_dt_lti}\\
    \mathcal{S}=\{x : Hx \leq h\},~
    \mathcal{U} = \{u:H_u u\leq h_u\}\label{eqn:rob_h_rep}\\
    w[k]\in\set{W}\big[\text{mod}(k,M)\big] \nonumber
\end{gather}
\end{subequations}
Substituting \eqref{eqn:rob_dt_lti},\eqref{eqn:rob_h_rep} into Def.~\ref{def:robmpre} generates:
\begin{align}
    Pre^M(\mathcal{S},\mathcal{W}) = 
    \left\{x\in\mathbb{R}^n: \exists ~u\in\mathbb{R}^m ~ s.t. ~\hat{H} 
    \begin{pmatrix} x \\ u \end{pmatrix} \leq \hat{h}\right\} \nonumber\\
    \hat{H} = \begin{bmatrix} 
    HA_{T_s} & HB_{T_s} \\ 
    HA^2_{T_s} & H(A_{T_s}B_{T_s}+B_{T_s})  \\
    \vdots & \vdots \\
    HA_{T_s}^M & H\left(\sum\limits_{n=1}^{M}{A_{T_s}^{M-n}B_{T_s}}\right)\\
    0 & H_u
    \end{bmatrix} \quad
    \hat{h}= \begin{bmatrix} \tilde{h}[0] \\ \tilde{h}[1] \\\vdots \\ \tilde{h}[M-1] \\ h_u \end{bmatrix} \label{eqn:rob_mpre_lti_b}
\end{align}
where
\begin{equation*}
    \tilde{h}_j[k]=\min_{w\in\set{W}\big[\text{mod}(k,M)\big]}(h_j-H_jw), ~\forall k\in\{0,\dots,M-1\}. \label{eqn:rob_lp}
\end{equation*}

Using \eqref{eqn:rob_mpre_lti_b} in Alg.~\ref{alg:rob_mcinf} is sufficient to calculate $\set{RC}^M_{\infty,T_s}$. 

\section{Nonlinear Example}
Consider the continuous-time, nonlinear system
\begin{equation}
    \dot{x}_r(t)=\sin\big(x_r(t)\big)+u(t). \label{eqn:nl_ex_fr}
\end{equation}
We model~\eqref{eqn:nl_ex_fr} using its linearization about $x_r=0$, $u_r=0$,
\begin{equation}
    \dot{x}_c=x_c(t)+u(t). \label{eqn:nl_ex_fc}
\end{equation}
The exact discretization of~\eqref{eqn:nl_ex_fc} with sampling time $T_s$ is
\begin{equation}
    x_d[k+1]=e^{T_s}x_d[k]+(e^{T_s}-1)u[k]. \label{eqn:nl_ex_fd}
\end{equation}
We assume the same input with a ZOH is applied to all models. Furthermore, all models are subject to the same state and input constraints:
\begin{equation}
    -1\leq x\leq 1,~-1\leq u\leq1  \label{eqn:nl_ex_cn}
\end{equation}

To find robust $M$-step hold control invariant sets for~\eqref{eqn:nl_ex_fd}, we must bound the error between the real system and the discrete model. First, we define the error between the real system and the continous model $e_c(t)=x_r(t)-x_c(t)$ with
\begin{align}
    \dot{e}_c(t)&=\dot{x}_r(t)-\dot{x}_c(t)\nonumber\\
    &=\sin\big(x_r(t)\big)-x_c(t).\nonumber
\end{align}
We use the Taylor series expansion of $\sin(x)$ and choose to bound the higher order terms as follows:
\begin{align}
    \dot{e}_c(t)&=x_r(t)-\frac{x_r(t)^3}{6}+\mathcal{O}\big(x_r(t)^5\big)-x_c(t) \nonumber\\
    &=e_c(t)-\frac{x_r(t)^3}{6}+\mathcal{O}\big(x_r(t)^5\big)\nonumber\\
    \big|\dot{e}_c(t)\big|&\leq e_c(t)+\frac{\big|x_r(t)\big|^3}{6}\nonumber\\
    &\leq e_c(t)+\frac{1}{6} \label{eqn:nl_ex_edot}
\end{align}
where $|x_r(t)|\leq1$ due to the state constraints~\eqref{eqn:nl_ex_cn}. We solve~\eqref{eqn:nl_ex_edot} to bound $e_c(t)$, assuming $e_c(0)=0$:
\begin{align}
    \big|e_c(t)\big|&\leq e_c(0)\exp(t)+\frac{1}{6}\big(\exp(t)-1\big) \nonumber\\
    &=\frac{1}{6}\big(\exp(t)-1\big) \label{eqn:nl_ex_ebound}
\end{align}
Since~\eqref{eqn:nl_ex_fd} was derived using exact discretization and~\eqref{eqn:nl_ex_ebound} is monotonically increasing, it is sufficient to query the bound on $e_c(t)$ every $kT_s$ to find the bound on the error between the real system and the discrete model $e_d[k]=x_r(kT_s)-x_d[t]$:
\begin{equation}
    \big|e_d[k]\big|\leq\frac{1}{6}\big(\exp(kT_s)-1\big) 
\end{equation}
This error bound is used as the bound on $w[k]$ in~\eqref{eqn:rob_dt_lti} to calculate $\set{RC}^M_{\infty,T_s}$ sets for~\eqref{eqn:nl_ex_fd}
\begin{equation}
    \big|w[k]\big|\leq e_d\big[\text{mod}(k,M)+1\big].
\end{equation}
These sets are shown for $T_s=0.1$ and $M=\{1, 5, 10\}$ in Figure~\ref{fig:nl_example}. 

\begin{figure}
    \centering
    \includegraphics[width=1\linewidth]{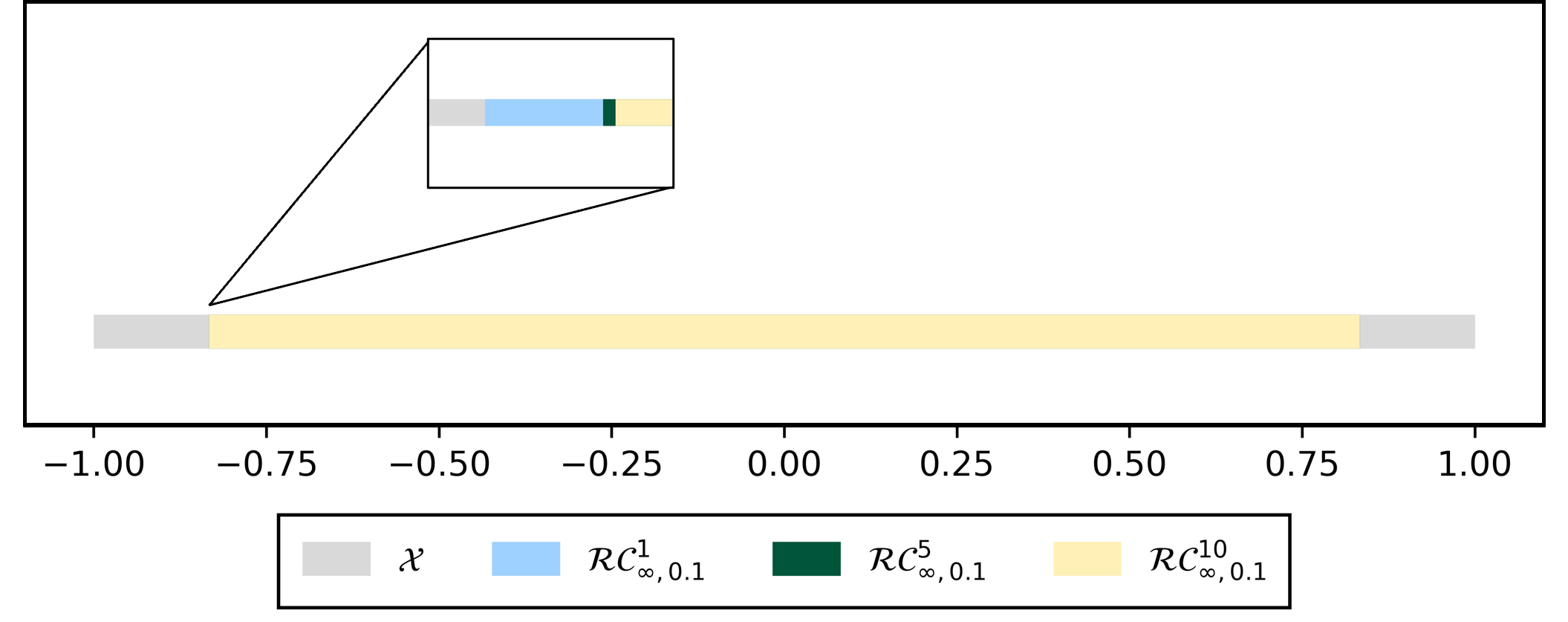}
    \caption{$\set{RC}^M_{\infty,T_s}$ for the nonlinear example.}
    \label{fig:nl_example}
\end{figure}

\section{Conclusion}\label{sec:conc}
We introduced $M$-step hold invariance 
for discrete-time systems, and demonstrated how it enables reasoning about constraint satisfaction for variable control sampling rates. We provided a framework for calculating $M$-step hold invariant sets and robustifying against modeling and discretization errors. Future work will leverage these sets to develop safe adaptive sampling controllers.


\bibliographystyle{ieeetr}
\bibliography{root}

\end{document}